\documentclass[prx,aps,superscriptaddress,showpacs,floatfix,tightenlines,twocolumn]{revtex4-2}
\usepackage[T1]{fontenc}
\usepackage[utf8]{inputenc}  
\usepackage[english,german]{babel}    
\usepackage{graphicx}
\usepackage{dcolumn}
\usepackage{bm}
\usepackage{amssymb}
\usepackage{bbm}
\usepackage{dsfont}
\usepackage{bbold}
\usepackage{mathbbol}
\usepackage{amsmath}
\usepackage{enumerate}
\usepackage{enumitem}
\usepackage{url}
\usepackage{amsthm}
\usepackage{layout}
\usepackage{epsfig}
\usepackage{graphicx}
\usepackage{epstopdf}
\usepackage{booktabs}
\usepackage{float}
\usepackage{subfigure}
\usepackage{lipsum}
\usepackage{mathrsfs}
\usepackage{booktabs}
\usepackage{blindtext}
\usepackage{appendix}
\usepackage[hyperindex,breaklinks]{hyperref}
\usepackage{enumerate}
\usepackage{enumitem}
\usepackage{makecell}
\usepackage{booktabs}
\usepackage{multirow}
\usepackage{graphicx}
\usepackage{tabularx}
\usepackage{array}
\usepackage{subfigure}
\allowdisplaybreaks[4]

\def\Tr{{\rm Tr}}
\newtheorem{theorem}{Theorem}

\newtheorem{definition}{Definition}
\newtheorem{proposition}{Proposition}

\allowdisplaybreaks[3]

\usepackage{xcolor}
\newcommand{\yaqi}[1]{\textcolor{black}{#1}}

\begin{document}

\title{A nonconvex entanglement monotone determining the characteristic length of entanglement distribution in continuous-variable quantum networks}

\author{Yaqi Zhao}
\affiliation{College of Mathematics, Taiyuan University of Technology, Taiyuan, 030024, China}

\author{Jinchuan Hou}
\affiliation{College of Mathematics, Taiyuan University of Technology, Taiyuan, 030024, China}

\author{Kan He}
\email{hekan@tyut.edu.cn}
\affiliation{College of Mathematics, Taiyuan University of Technology, Taiyuan, 030024, China}

\author{Nicolò~Lo~Piparo}
\email{nicopale@gmail.com}
\affiliation{Okinawa Institute of Science and Technology Graduate University,
1919-1 Tancha, Onna-son, Okinawa 904-0495, Japan}

\author{Xiangyi Meng}
\email{xmenggroup@gmail.com}
\affiliation{Department of Physics, Applied Physics, and Astronomy, Rensselaer Polytechnic Institute, Troy, New York 12180, USA}
\affiliation{Network Science and Technology Center, Rensselaer Polytechnic Institute, Troy, New York 12180, USA}

\date{\today }
\begin{abstract}

Quantum networks (QNs) promise to enhance the performance of various quantum technologies in the near future by distributing entangled states over long distances. The first step towards this is to develop novel entanglement measures that are both informative and computationally tractable at large scales. While numerous such entanglement measures exist for discrete-variable (DV) systems, a comprehensive exploration for experimentally preferred continuous-variable (CV) systems is lacking. Here, we introduce a class of CV entanglement measures, among which we identify a nonconvex entanglement monotone---the ratio negativity, which possesses a simple, scalable form that determines the exponential decay of optimal entanglement swapping on a chain of pure Gaussian states. This characterization opens avenues for leveraging statistical physics tools to analyze swapping-protocol-based CV QNs.

\end{abstract}

\pacs{03.67.Mn, 03.65.Ud, 03.67.-a}
\maketitle

\section{Introduction}

Entanglement is a fundamental resource for several quantum information tasks, such as quantum communication~\cite{quantum_comuputation_2007,Repeater_2007,QComm,Quantumcomm,Repeater_Koji_2015,Repeater_Koji_2023,key_distribution_2014,key_distribution_2023,Swapping_experiement2022}, quantum computation~\cite{Qcomp2,Quantum_comp2,Quantum_comp3}, and quantum sensing and imaging~\cite{QSensing,QImaging2,QImaging3}. 
These tasks can be accomplished by using quantum networks (QNs), in which entangled states are created and distributed between remote parties (nodes) through local operations and classical communication (LOCC)~\cite{entanglement_distribution1997}.
At large scales, it is important to explore the global features of entanglement distribution on QNs to see whether such features can be mapped to known statistical physics models.
Towards this goal, several attempts have already been made. For instance, in 2007, an entanglement distribution scheme, known as entanglement percolation~\cite{Percolation_CEP2007}, showed a way to map the distribution of partially entangled pure states to bond percolation theory~\cite{Percolation1980}. Since then, research in entanglement percolation has expanded to include more statistical physics and complex network models~\cite{percolation2008,percolation2009,percolation_complex2009,Percolation_mixed2010,Percolation_multi2010,Percolation_cluster2011,Pecplation_ConPT2021,Percolation_LargeScale2022,Percolation_review2023,path-percolation_mhrk24,path-percolation_UVflower2024,Percolation_QMemo2024}, providing valuable insights on understanding, enhancing, and optimizing the efficiency and reliability of entanglement distribution schemes.

\yaqi{A fundamental quantity often used to characterize the behavior of  statistical physics models is the characteristic length $\xi$, which describes the exponential decay of correlation over a long distance $l$~\cite{characteristic_length1966,characteristic_length1967}.  Similarly, when considering an \emph{entanglement transmission} scheme---the distribution of entanglement between \emph{two} nodes~\cite{DET2023}---we can define a characteristic length $\xi$ to describe the exponential decay of entanglement transmission between two nodes within a QN. Specifically, for two distant nodes, $S$ and $T$, separated by a distance $l$, we define \yaqi{$\xi:= -\lim\limits_{l\rightarrow\infty}\left[{l}/{\ln{E(\rho_{ST}})}\right]$}, where \yaqi{$E(\rho_{{ST}})$ denotes the final entanglement established between $S$ and $T$ as quantified by some entanglement measure $E$.} This characteristic length $\xi$ can vary depending on the entanglement transmission scheme and the measure used. Therefore, determining
$\xi$ can often reveal how fast entanglement can deteriorate in the specific scheme.} 

In this work, we focus on entanglement transmission schemes for a one-dimensional (1D) QN chain. These schemes primarily rely on entanglement swapping protocols~\cite{swapping1993,swapping1999,vanLoock2002,swapping_CV2004,swapping_CV2005_1,swapping_Gaussian2011,Takeda2015,swapping_optimal2022} as their building blocks to distribute entangled pairs between two remote nodes of the QN.
Specifically, entanglement swapping uses two initial entangled pairs (e.g., $S$--$R$ and $R$--$T$) to create a new entangled pair directly between $S$ and $T$ through local operations at a ``relay'' node $R$~\cite{swapping1993,swapping1999,vanLoock2002,swapping_CV2004,swapping_CV2005_1,swapping_Gaussian2011,Takeda2015,swapping_optimal2022}. 
The resulting entanglement $E(\rho_{ST})$ depends on the entanglement of the two initial pairs, $E(\rho_{SR})$ and $E(\rho_{RT})$. 
If we can find a measure $E\in[0,1]$ satisfying the \emph{multiplicative property}, $E(\rho_{ST}) = E(\rho_{SR})\cdot E(\rho_{RT})$ \yaqi{for arbitrary $\rho_{SR}$ and $\rho_{RT}$}, then consider a 1D chain comprising $l$ identical states ($\rho=\rho_{SR_1}=\dots=\rho_{R_{l-1}T}$) as shown in Fig.~\ref{fig-seriesN}, the final entanglement created between $S$ and $T$ through entanglement swapping will simply be $E(\rho_{ST})= E(\rho)^l$. This means that the value $l/\ln E(\rho_{ST})$ is independent of $l$, and hence the characteristic length readily reads {$\xi\propto-\left(\ln E\right)^{-1}$} \yaqi{(the proportionality accounts for the fact that any $E^{\alpha}$ would also satisfy the multiplicative property; thus, the specific entanglement measure used to define $\xi$ should depend on the task at hand). }

\begin{figure}[t!]
    \centering
    \includegraphics[width=200pt]{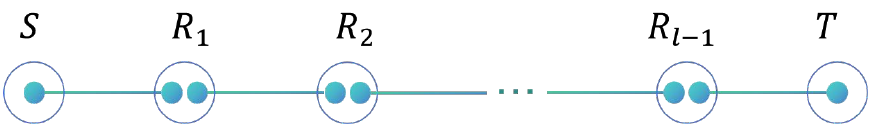}
    \caption{A general one-dimensional quantum network. Each node (circle) contains qubits, qudits, or optical modes (green dots), and links represent entanglement between different nodes.}
    \label{fig-seriesN}
\end{figure}

Here, we \yaqi{derive a necessary and sufficient condition for the existence of the multiplicative property for \emph{any} swapping protocol described by $E(\rho_{ST}) = g(E(\rho_{SR}), E(\rho_{RT}))$ \yaqi{with some generic function $g$}.
We then verify that several entanglement measures that have been proposed, such as the singlet conversion probability~\cite{SCP1999}, the concurrence~\cite{Concurrence2001,Concurrence_convex_roof2003_2,Concurrence2013}, and the more general $G$-concurrence~\cite{G_concurrence2005}, all exhibit this multiplicative property for various known entanglement swapping protocols using discrete-variable (DV) systems (Table~\ref{table-measures}). 
However, such an entanglement measure for continuous-variable (CV) systems has not yet been identified. 
Our main contribution is the discovery that, for two-mode squeezed vacuum states (TMSVSs)---a common class of entangled pure CV states---the corresponding {Gaussian swapping protocol}~\cite{vanLoock2002} does possess such a measure, which we call the \emph{ratio negativity} $\chi_{\mathcal{N}}\in[0,1]$, that satisfies $\chi_{\mathcal{N}}^{\alpha}(\rho_{ST})=\chi_{\mathcal{N}}^{\alpha}(\rho_{SR})\chi_{\mathcal{N}}^{\alpha}(\rho_{RT})$ with $\alpha>0$ (we also refer to $\chi_{\mathcal{N}}^{\alpha}$ as the \emph{$\alpha$-ratio negativity} for consistency).
Similar to how the logarithmic negativity~\cite{Plenio2005} can serve as an upper bound for distillable entanglement, the ($\alpha$-)ratio negativity $\chi_{\mathcal{N}}^{\alpha}$, also \yaqi{derived} from the negativity $\mathcal{N}$~\cite{Vidal2002},  now adds another operational significance to the negativity family---showing that the ($\alpha$-)ratio negativity could be a suitable figure of merit for quantifying entanglement in 1D QNs for CV systems.}

\begin{table*}[t!]
    \caption{\yaqi{Characteristic length $\xi$ of different  entanglement transmission schemes and resources in 1D quantum network (QN).}}
    \begin{tabular}{c|c|c}
		\hline\hline
		\yaqi{QN} resource state & \yaqi{1D entanglement transmission scheme}  & \yaqi{Entanglement measure $(=e^{-1/\xi})$
        } \\
		\hline
            2-qubit DV pure state &  singlet conversion LOCC & singlet conversion probability~\cite{SCP1999}\\
		  2-qubit DV pure state & optimal DV deterministic swapping  for qubits~\cite{percolation2008} & concurrence~\cite{Hill1997}\\
		2-qudit DV pure state & optimal DV deterministic swapping for qudits~\cite{DET2023} & $G$-concurrence~\cite{G_concurrence2005}\\
            TMSVS CV pure state & optimal Gaussian deterministic swapping~\cite{vanLoock2002} & ratio negativity \yaqi{(present work)}\\

		\hline\hline
	\end{tabular}\label{table-measures}
\end{table*}

The $\alpha$-ratio negativity exhibits several key properties. First, it is a nonconvex entanglement monotone for $0< \alpha \leqslant 1$, exhibiting nontrivial characteristics~\cite{EM2000}.
Moreover, the $\alpha$-ratio negativity shows \emph{monogamy}, meaning that as two parties become more strongly entangled, their ability to share $\alpha$-ratio negativity with others decreases~\cite{Monogamy2000}.
We find that when $\alpha\geqslant \log_{3\left(\sqrt{2}-1\right)}2\approx 3.191$, for 
certain classes of multipartite pure states, the $\alpha$-ratio negativity satisfies the CKW inequality proposed by Coffman, Kundu, and Wootters~\cite{Monogamy2000}, a signature of monogamy. 
This result can be readily extended to mixed states within the same multipartite classes by employing the convex-roof extension~\cite{Kim2018}.
Our findings indicate that the $\alpha$-ratio negativity is inherently monogamous in several cases of multipartite settings.

\yaqi{Given an unknown bipartite state $\rho$, a natural question arises: how can one experimentally measure the ratio negativity $\chi_{\mathcal{N}}(\rho)$? Given that the ratio negativity is a variant of the negativity $\mathcal{N}$, any methodology, whether direct or indirect, capable of determining $\mathcal{N}(\rho)$, can be harnessed to derive the ratio negativity. In practice, tomography techniques~\cite{Tomography_qudit2002,Tomography_entanglement_CV2002,quantum-tomography,quantum-state-estimation,Tomography_review2009,Tomography_Gaussian2009,Tomography_Gaussian2010,Tomography_full2016,CV_moment_covariance2017,Tomography_process2020,Tomography_higher_dim_adaptive2020,Tomography_Gaussian2024} are typically used to reconstruct $\rho$, enabling the subsequent calculation of the negativity. Alternatively, machine learning methodologies~\cite{Tomography_measurement_learning2018} can be directly applied to measure $\mathcal{N}(\rho)$.}

\section{Preliminaries}\label{sec-pre}

In this Section, we introduce the basic concepts and definitions that will help the reader to better understand our results.

\emph{CV system.---}A CV system, generally described by $N$ bosonic modes ($N$-mode), is associated with a tensor-product Hilbert space ${\mathcal H}=\otimes_{k=1}^{N}{\mathcal H}_{k}$, where ${\mathcal H}_{k}$ is an infinite-dimensional complex Hilbert space that corresponds to the $k$th mode.
For each Hilbert subspace $\mathcal{H}_k$, we denote $\{|n\rangle\}_{n=0}^{\infty}$ as the Fock basis.
The sets of density matrices of all states and of all pure states defined in $\mathcal H$ are denoted by $\mathbb{S}(\mathcal H)$ and $\mathbb{PS}(\mathcal H)$, respectively.

\emph{Entanglement measures and monotones.---}Assume that $\mathcal H_{AB}=\mathcal{H}_{A}\otimes\mathcal{H}_{B}$ is a bipartite system.
The authors of Refs.~\cite{EM2000,measure2007} proposed the axiomatic theory of entanglement measures~\cite{Axiomatic_def_measure2009}: A function $E$: $\mathbb{S}(\mathcal{H}_{AB})\rightarrow \mathbb{R}_+$ is called an \emph{entanglement measure} iff (i) it vanishes for all separable states of $\mathcal{H}_{AB}$, and (ii) it does not increase under any LOCC, i.e.,
\begin{eqnarray}
\label{eq_measure_general}
    E(\rho)\geqslant E(\Phi(\rho))
\end{eqnarray}
for any LOCC map $\Phi$ that is completely positive and trace preserving~\cite{Axiomatic_def_measure2009}.
In general, $\Phi(\rho)$ is given by an ensemble $\{p_i,\rho_i'\}_i$ made by all possible outcome states $\rho_i'$  with corresponding probabilities $p_i$, meaning that $\rho$ can be transformed by LOCC into $\rho_i'$ with probability $p_i$  
($\sum_i p_i=1$).  Eq.~\eqref{eq_measure_general} then becomes
\begin{eqnarray}
\label{eq_measure}
    E(\rho)\geqslant E\left(\sum_i p_i \rho_i'\right).
\end{eqnarray}
An entanglement measure that satisfies the following condition 
\begin{eqnarray}
\label{eq_monotone}
  E(\rho)\geqslant\sum_j q_j E(\sigma_j')
\end{eqnarray}
for any output ensemble $\{q_j,\sigma_j'\}_j$ is called an \emph{entanglement monotone}~\cite{Axiomatic_def_measure2009}. Note that for the special case of a single outcome $\sigma_1'=\sum_i p_i \rho_i'$ with a probability of unity ($q_1=1$), Eq.~\eqref{eq_monotone} reduces to Eq.~\eqref{eq_measure},
implying that entanglement monotone is a stronger condition than entanglement measure. 
Most entanglement monotones also meet a  stricter requirement, being \emph{convex functions} satisfying
\begin{eqnarray}
\label{eq_convex}
  \sum_j q_j E(\rho_j)\geqslant E\left(\sum_j q_j \rho_j\right).
\end{eqnarray}
The convexity property is, however, not a necessary condition for an entanglement measure to be a monotone~\cite{Plenio2005}.
Hence, one should regard a nonconvex entanglement measure that satisfies the property in Eq.~\eqref{eq_monotone} as a \emph{nonconvex entanglement monotone}.

Two widely used entanglement measures for CV systems are the \emph{negativity} and the \emph{logarithmic negativity}~\cite{Vidal2002,Plenio2005}. Neither increases with LOCC, or with more general positive partial transpose (PPT) preserving maps~\cite{Plenio2005}.
Both the {negativity} and the {logarithmic negativity} are also entanglement monotones~\cite{Plenio2005}. However, while the negativity is a convex function, the logarithmic negativity does not satisfy Eq.~\eqref{eq_convex}.
Specifically, given a matrix $W$ expressed in the Fock basis of $\mathcal{H}_{AB}$, its partial transpose $W^{T_A}$ on subsystem $A$ is defined as $\langle j,k|W^{T_{A}}|l,s\rangle=\langle l,k|W|j,s\rangle$. 
The negativity $\mathcal{N}$ and the logarithmic negativity $E_{\mathcal{N}}$~\cite{Vidal2002} of a (pure or mixed)  state  $\rho\in\mathcal{H}_{AB}$ are defined as
  \begin{eqnarray}
    \label{EM}
     \mathcal{N}(\rho):=\frac{\|{\rho}^{T_{A}}\|_{1}-1}{2},\  E_{\mathcal{N}}(\rho):=\log_2{\|{\rho}^{T_{A}}\|_{1}},
  \end{eqnarray}
where $\|W\|_{1}:=\Tr\sqrt{W^{\dag}W}$ is the trace norm of $W$. Given a representation of the state $\rho$ in the Fock basis, the negativity and the logarithmic negativity are both functions of the sum of the eigenvalues of $\sqrt{\left(\rho^{T_{A}}\right)^{\dag}\rho^{T_{A}}}$. Since $\mathcal{N}(\rho)$ is the absolute value of the sum of the negative eigenvalues of $\rho^{T_A}$~\cite{Vidal2002}, both $\mathcal{N}(\rho)$ and $E_{\mathcal{N}}(\rho)$ vanish iff $\rho$ is a PPT 
state, meaning that the partial transpose of $\rho$ is positive semidefinite. 
When $\rho$ is a Gaussian state (Appendix~\ref{sec-AP-CM}), $\mathcal{N}(\rho)$ and $E_{\mathcal{N}}(\rho)$ can be entirely determined by the covariant matrix (Appendix~\ref{sec-AP-CM})~\cite{JB,MTJ,RT,GAF, GF2005,ST2017}.
In particular, if $\rho$ is $(1+M)$-mode or $(N+M)$-mode bisymmetric~\cite{NMmode_bisymmetric_Gaussian2005} (i.e.,~locally invariant under the exchange of any two modes within any one of two 
subsystems), then $E_{\mathcal{N}}(\rho)=0$ is a necessary and sufficient condition for the separability of $\rho$~\cite{1Mmode_Gaussian2001,GAF,NMmode_bisymmetric_Gaussian2005,AI2007}.

One potential drawback of the negativity and the logarithmic negativity measures is that neither of them is bounded, owing to the unbounded nature of the trace norm $\|{\rho}^{T_{A}}\|_1$. For instance, for a TMSVS $|\psi^{r}\rangle$, we have $\mathcal{N}(|\psi^{r}\rangle\langle\psi^{r}|)=(\exp{2r}-1)/2$ and $E_{\mathcal{N}}(|\psi^{r}\rangle\langle\psi^{r}|)=2r$, both of which tend to infinity as $r\rightarrow \infty$.

A bounded entanglement monotone is the \emph{concurrence}~\cite{Hill1997,Concurrence2001,Concurrence_convex_roof2003_2,Concurrence2013}, which for a pure state $|\psi\rangle\langle\psi| \in \mathbb{PS}(\mathcal{H}_{AB})$ is defined by $C(|\psi\rangle\langle\psi|):=\sqrt{2\left(1-\Tr(\rho_{A}^{2})\right)}$ where $\rho_{A}=\Tr_{B}(|\psi\rangle\langle\psi|)$. For a mixed state $\rho\in \mathbb{S}(\mathcal{H}_{AB})$, its concurrence is defined by means of the generalized convex roof construction~\cite{Concurrence_convex_roof2003_2,Concurrence2013}:
\begin{equation}
    C(\rho):=\inf\limits_{\{p_i,|\psi_i\rangle\}}\sum\limits_{i}p_i C\left(|\psi_i\rangle\langle \psi_i|\right),
\end{equation}
where the infimum is taken over all possible pure-state ensembles $\{p_i,|\psi_i\rangle\}$, leading to $\rho=\sum\nolimits_i p_{i}|\varphi_i\rangle\langle\varphi_i|$. Therefore, given a generic mixed state, unlike the negativity and the logarithmic negativity, the concurrence cannot be directly computed from the representation of the state in the Fock basis,  due to the fact that the number of possible ensembles $\{p_i,|\psi_i\rangle\}$ is exponentially large.

The definition of concurrence can be extended to another entanglement monotone~\cite{G_concurrence2005}.
Considering a $(d\times d)$-dimensional bipartite pure state $|\psi\rangle$ with Schmidt decomposition $|\psi\rangle=\sum_{j=0}^{d-1}\sqrt{\lambda_j}|jj\rangle$ where $\lambda_j>0$ and $\sum_j\lambda_j=1$, its entanglement can be quantified by the entanglement monotone called \emph{$G$-concurrence}~\cite{G_concurrence2005}, $C_{\mathrm{G}}(|\psi\rangle\langle\psi|):=d(\lambda_0\lambda_1\dots\lambda_{d-1})^{1/d}$. For a mixed state, the $G$-concurrence can be similarly defined using the generalized convex roof construction~\cite{Concurrence_convex_roof2003_2}. When $d=2$, the $G$-concurrence reduces to the concurrence measure.

\emph{Entanglement monogamy.---}The concept of entanglement monogamy states that 
entanglement has an ``exclusive'' property such that in a multipartite quantum system, the more strongly two parties are entangled, the less they can share entanglement with others. This principle is illustrated by analyzing certain entanglement measures that adhere to specific inequalities.
Initially, an entanglement measure $E$ was considered monogamous if the CKW inequality~\cite{Monogamy2000}
\begin{eqnarray}\label{ineq-CKW}
    E(\rho_{A|B_1B_2\cdots B_{N-1}})\geqslant \sum_{n=1}^{N-1}E(\rho_{AB_n}),
\end{eqnarray}
holds for every $N$-party state $\rho_{AB_1B_2\cdots B_{N-1}}\in\mathbb{S}(\mathcal{H}_{AB_1B_2\cdots B_{N-1}})$, where $\rho_{A|B_1B_2\cdots B_{N-1}}$ denotes the bipartite state under the partition $A$ and $B_1B_2\cdots B_{N-1}$, and $\rho_{AB_{n}}=\Tr_{B_1\cdots B_{n-1}B_{n+1}\cdots B_{N-1}}(\rho_{AB_1B_2\cdots B_{N-1}})$ is the reduced state of $\rho_{AB_1B_2\cdots B_{N-1}}$ in subsystem $AB_n$. 
Subsequent studies reveal that most entanglement measures do not satisfy Eq.~\eqref{ineq-CKW}~\cite{Monogamy2000,Monogamy2006,Monogamy_ou2007,Monogamy_kim2009,Monogamy2014,Monogamy_Bai2014,Monogamy_Choi2015,Monogamy2015,Monogamy_zhu2015}.
Later, Gour and Guo propose an extended definition of monogamy based on tripartite systems~\cite{monogamy2018}: given three parties $A$, $B$ and $C$, a bipartite entanglement measure $E$ is said to be monogamous, if for every tripartite systems in an Hilbert space $\mathcal{H}_{ABC}$, every $\rho_{ABC}\in\mathbb{S}(\mathcal{H}_{ABC})$ that satisfies $E(\rho_{A|BC})=E(\rho_{AB})$, we have $E(\rho_{AC})=0$.
According to Ref.~\cite{monogamy2018}, this definition implies that a continuous entanglement measure $E$ is monogamous in a tripartite system $\mathcal{H}_{ABC}$ with fixed finite dimension iff there exists some $\alpha>0$ such that $E^{\alpha}(\rho_{A|BC})\geqslant  E^{\alpha}(\rho_{AB})+E^{\alpha}(\rho_{AC})$ holds for all $\rho_{ABC}\in\mathbb{S}(\mathcal{H}_{ABC})$.
The natural extension of the definition in $N$-party systems is that if for every $N$-party state $\rho_{AB_1B_2\cdots B_{N-1}}$ with fixed finite dimension that satisfies
\begin{eqnarray}\label{eq-monogamy-generalN}
    E(\rho_{A|B_1B_2\cdots B_{N-1}})= E(\rho_{AB_k}),
\end{eqnarray}
we have $E(\rho_{AB_n})=0$ for all $n\neq k$. This means that if there exists some $0<\alpha<\infty$ such that 
\begin{eqnarray}\label{ineq-monogamy-generalN}
    E^{\alpha}(\rho_{A|B_1B_2\cdots B_{N-1}})
   \geqslant\sum_{n=1}^{N-1}E^{\alpha}(\rho_{AB_n})
\end{eqnarray}
holds for all $N$-party state $\rho_{AB_1B_2\cdots B_{N-1}}$, then $E$ satisfies the \emph{extended monogamy}.

\section{Ratio negativity}\label{sec_XN}

To introduce our novel entanglement measures, we start by presenting a class of negativity-based entanglement measures, which we call the \emph{f-negativity}.

\begin{definition}\label{def1}
($f$-negativity) Given a real function $f$ which satisfies the following two conditions:\\
    (i) $f(0)=0$,\\
    (ii) $f$ is strictly increasing in $[0,+\infty)$,\\
    we define the function-based negativity generated by $f$ for any bipartite state $\rho$ as
\begin{eqnarray}
   f_{\mathcal{N}}(\rho):= f(\mathcal{N}(\rho)).
\end{eqnarray}
We call $f_\mathcal{N}$ the $f$-negativity. 
\end{definition}

For $f(x):=x$, the $f$-negativity reduces to the conventional negativity measure. For $f(x):=\ln{(2x+1)}$, the $f$-negativity reduces to the logarithmic negativity.
By definition, $f_{\mathcal N}(\rho)=0$ iff $\mathcal N(\rho)=0$. Furthermore, the $f$-negativity has the following basic properties: 

1) \emph{Monotonicity.} It follows the same ordering with the negativity in the set of all bipartite states, i.e., 
 \begin{eqnarray}\label{equivalence}
     f_{\mathcal{N}}(\rho_{1})\geqslant f_{\mathcal{N}}(\rho_{2})\rightleftharpoons \mathcal{N}(\rho_{1})\geqslant\mathcal{N}(\rho_{2}).
 \end{eqnarray}

2) \emph{Semipositivity.} The inequality $f_{\mathcal{N}}(\rho)\geqslant0$ holds for any bipartite state $\rho$.

These properties further lead to the following results (see Appendix~\ref{sec-AP-properties} for proofs):

3) \emph{$f_{\mathcal{N}}$ is an entanglement measure.} 

4) \emph{$f_{\mathcal{N}}$ is convex when $f$ is convex.}

5) \emph{$f_{\mathcal{N}}$ is an  entanglement monotone when $f$ is concave.}

\begin{definition}
(Ratio negativity) We define the $f$-negativity denoted by the \emph{ratio negativity} $\chi_{\mathcal{N}}$, given by the concave function $f(x):=x/\left(x+1\right)$ ($x\geqslant0$),
\begin{eqnarray}
    \chi_{\mathcal{N}}{(\rho)}:=\frac{\|\rho^{T_A}\|_1-1}{\|\rho^{T_A}\|_1+1}.
\end{eqnarray}
\end{definition}
\begin{definition} ($\alpha$-ratio negativity)
We denote the $\alpha$-th power ($\alpha>0$) of the ratio negativity as the \emph{$\alpha$-ratio negativity}, $\chi^{\alpha}_{\mathcal{N}}$, which is also an $f$-negativity given by $f(x):=\left[x/\left(x+1\right)\right]^\alpha$. 
\end{definition}

The ratio negativity $\chi_{\mathcal{N}}$ is an entanglement monotone in the range $[0,1]$ that is nonconvex. More broadly, for any $0< \alpha \leqslant 1$, the $\alpha$-ratio negativity, $\chi^{\alpha}_{\mathcal{N}}$, also constitutes a nonconvex entanglement monotone. 

For pure states, which are frequently employed  in entanglement percolation theories~\cite{Percolation_review2023}, the ratio negativity has a simple form.
Consider a bipartite pure state whose Schmidt decomposition is given by $|\psi\rangle=\mathop{\Sigma}_{\alpha}\sqrt{\lambda_{\alpha}}|s_{\alpha}l_{\alpha}\rangle$ for $\lambda_{\alpha}>0,\ \mathop{\Sigma}_{\alpha}\lambda_{\alpha}=1$. We have
  \begin{eqnarray}\label{eq-pure}
         \chi_{\mathcal N}(|\psi\rangle\langle\psi|)=
                 \frac{(\mathop{\sum}\limits_{\alpha}\sqrt{\lambda_{\alpha}})^2-1}
                      {(\mathop{\sum}\limits_{\alpha}\sqrt{\lambda_{\alpha}})^2+1}.
  \end{eqnarray}

{A commonly used pure state is the two-mode squeezed vacuum state (TMSVS), a pure $(1+1)$-mode Gaussian state characterized by a squeezing parameter $r>0$:
\begin{eqnarray}
    \label{TMSVS}
    |\psi^{r}\rangle=\sqrt{1-\chi^{2}}\sum_{n=0}^{+\infty}\chi^{n}|nn\rangle,\ \chi\equiv \tanh{r}.
  \end{eqnarray}}
By Eq.~\eqref{eq-pure}, its ratio negativity is simply
  \begin{eqnarray}\label{eq-EPR}
    \label{ER-r}
    \chi_{\mathcal N}(|\psi^{r}\rangle\langle\psi^{r}|)=\tanh r=\chi,
  \end{eqnarray}
which is closely related to the explicit representation of the TMSVS.
This simple representation plays an important role in characterizing 1D QNs, as we will demonstrate later in this work.

The ratio negativity satisfies the following inequalities: 

1) \emph{Subadditivity.---}Given Hilbert spaces $\mathcal H_{A}=\mathcal H_{A_{1}}\otimes\mathcal H_{A_{2}}$, $\mathcal H_{B}=\mathcal H_{B_{1}}\otimes\mathcal H_{B_{2}}$, $\mathcal H_{A_{k}B_{k}}=\mathcal H_{A_{k}}\otimes\mathcal H_{B_{k}}$, and a state $\rho_{k}$ of $\mathcal{H}_{A_kB_k}$ ($k=1,2$).
Then the state $\rho=\rho_1\otimes\rho_2$ of bipartite system $AB$ satisfies the subadditivity
  \begin{eqnarray}
   \label{tensor-2}
   \begin{aligned}
      \chi_{\mathcal N}(\rho)=\frac{\chi_{\mathcal N}(\rho_{1})+\chi_{\mathcal N}(\rho_{2})}{1+\chi_{\mathcal N}(\rho_{1})\chi_{\mathcal N}(\rho_{2})}\\
      \leqslant \chi_{\mathcal N}(\rho_{1})+\chi_{\mathcal N}(\rho_{2}),\ \
   \end{aligned}
  \end{eqnarray}
where the equality is true only when $\rho_{1}$ or $\rho_{2}$ is a PPT state (Appendix~\ref{sec-AP-tensor}).

2) \emph{Monogamy.---}We investigate the monogamy of the $\alpha$-ratio negativity.
 We find that while the ratio negativity ($\alpha=1$) violates the CKW inequality (Appendix~\ref{sec-AP-CKW}), for larger $\alpha$ the inequality could be satisfied.
To this end, we introduce the convex-roof extended negativity $\widetilde{\mathcal{N}}$~\cite{Negativity_convex_roof2003}:
\begin{eqnarray}\label{eq-CREN}
    \widetilde{\mathcal{N}}(\rho):=\inf\limits_{\{p_{i},|\varphi_i\rangle\}}\sum_{i}p_i\mathcal{N}(|\varphi_i\rangle\langle\varphi_i|)
\end{eqnarray}
for a bipartite state $\rho=\sum\nolimits_i p_{i}|\varphi_i\rangle\langle\varphi_i|$,
where the minimum is taken over all possible ensembles of pure states that construct $\rho$.
This construction allows us to generalize the $\alpha$-ratio negativity $\chi_{\mathcal{N}}^{\alpha}$ to a new entanglement measure:
\begin{definition}
    (The convex-roof extended $\alpha$-ratio negativity) For every bipartite state $\rho_{AB}$ of $\mathcal{H}_{AB}$, we define the ratio convex-roof extended $\alpha$-ratio negativity $\chi_{\widetilde{\mathcal{N}}}^{\alpha}$ for $\rho_{AB}$ as
\begin{eqnarray}
    \chi_{\widetilde{\mathcal{N}}}^{\alpha}(\rho_{AB})= f\left(\widetilde{\mathcal{N}}(\rho_{AB})\right)
\end{eqnarray}
with $f(x):=\left[x/(x+1)\right]^{\alpha}$.
\end{definition}
Note that the inequality $\widetilde{\mathcal{N}}(\rho)\geqslant\mathcal{N}(\rho)$ always holds, with the equality holding if $\rho$ is a pure state~\cite{Negativity_convex_roof2003}. 
Due to the monotonicity of $f$, this relationship extends directly to $\chi_{\widetilde{\mathcal{N}}}^{\alpha}(\rho) \geqslant\chi_{\mathcal{N}}^{\alpha}(\rho)$. Thus, this allows us to first consider the monogamy of $\chi_{\widetilde{\mathcal{N}}}^{\alpha}$ for {all} states of certain systems; then, we reduce the result to pure states, showing the monogamy of $\chi_{\mathcal{N}}^{\alpha}$ for the same systems.

We focus on the following $N$-party systems, denoted as $\mathcal{H}_{AB_1B_2\cdots B_{N-1}}$, consisting of parties $A,B_1,\cdots,B_{N-1}$ and satisfying one of these conditions:

1) The dimension of $\mathcal{H}_{AB_1B_2\cdots B_{N-1}}$ is $\overbrace{2\times 2\times\dots\times 2}^N$, i.e.,~each party is a qubit system.

2) $N=3$ and the dimension is $2\times 2\times 3$ or $2\times 2\times 2^m$ for arbitrary positive integer $m$.

We find that when $\alpha\geqslant \log_{3\left(\sqrt{2}-1\right)}2\approx 3.191$, $\chi_{\widetilde{\mathcal{N}}}^{\alpha}$ satisfies the CKW inequality:
\begin{eqnarray}\label{ineq-XNa1}
   \chi_{\widetilde{\mathcal{N}}}^{\alpha}(\rho_{A|B_1 B_2 \cdots B_{N-1}})\
    \geqslant\sum_{n}\chi_{\widetilde{\mathcal{N}}}^{\alpha}(\rho_{AB_n})
\end{eqnarray}
for every state $\rho_{AB_1B_2\cdots B_{N-1}}$ of $\mathcal{H}_{AB_1B_2\cdots B_{N-1}}$ (Appendix~\ref{sec-AP-extended-ratio-monogamy}). Here, the vertical indicates the bipartite split across which
the (bipartite) entanglement is measured and $\rho_{AB_n}$ is the reduced state of $\rho_{AB_1B_2\cdots B_{N-1}}$ in subsystem $AB_n$. Finally, reducing to pure states, we further prove that the same monogamy manifests in the $\alpha$-ratio negativity (not convex-roof extended) for pure states of the same systems (Appendix~\ref{sec-AP-ratio-monogamy}).

Notably, the monogamy definition itself can be  generalized, encompassing ``extended monogamy~\cite{monogamy2018},'' where an entanglement measure is considered extended-monogamous so long as some power of it satisfies the CKW inequality [Eq.~\eqref{ineq-monogamy-generalN}]. In this light, the ratio negativity exhibits extended monogamy, since its powered variant, the $\alpha$-ratio negativity, already demonstrates monogamy.

\yaqi{\section{Characteristic length of entanglement transmission}}

We start by defining the \yaqi{\emph{the characteristic length of an entanglement transmission scheme}} in a 1D QN as follows:
\yaqi{\begin{definition}(\yaqi{the characteristic length})
    For a given swapping-based entanglement transmission scheme in a 1D chain of $l$ identical states (see Fig.~\ref{fig-seriesN}), we define the characteristic length $\xi$  as the ratio of $l$ to the final entanglement between $S$ and $T$ through the scheme,
    \begin{eqnarray}
    \label{eq_characteristic_measure}      
    \xi:= -\lim\limits_{l\rightarrow\infty}\frac{l}{\ln{E(\rho_{ST})}},
    \end{eqnarray}
    where $E(\rho_{ST})$ denotes the final entanglement established between $S$ and $T$, quantified by an entanglement measure $E$.
\end{definition}}

\yaqi{If the value $-l/\ln{E(\rho_{ST})}$ in Eq.~\eqref{eq_characteristic_measure} is independent of $l$, then the characteristic length $\xi$ can be directly determined by $E$ independent of $l$ already before taking $l\rightarrow\infty$. 
This can be achieved by a class of entanglement measures that satisfy the multiplicative property. 
In the following, we discuss what entanglement measures belong to this class.
}

\begin{figure}[t!]
  \centering
  \subfigure[]{
   \includegraphics[width=90pt]{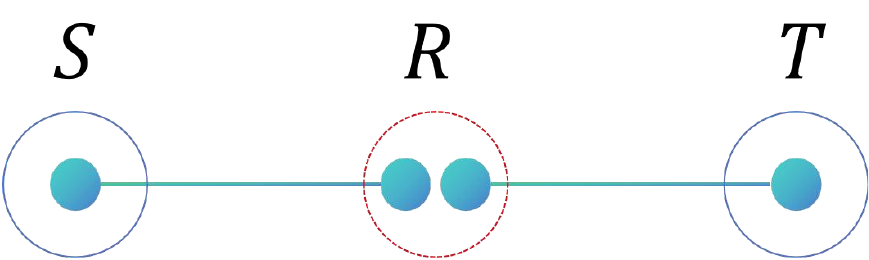}
  }
  \hspace{0pt}
  \subfigure[]{
   \includegraphics[width=90pt]{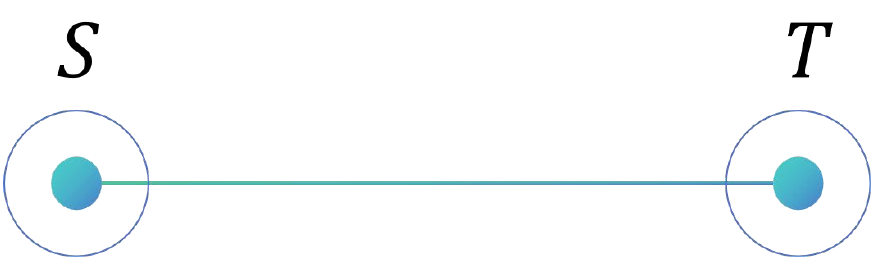}
  }
  \caption{\label{fig-series1} Series rule. (a) The simplest 1D network consisting of three nodes, $S$ (source), $R$ (relay), and $T$ (target). (b) The resulting link between $S$ and $T$ after applying the swapping protocol.}
\end{figure}

We begin with an example of simple bond percolation in a classical network~\cite{bond_percolation1980,cross-probab-square_k80}. Consider a simple 1D (series) network shown in Fig.~\ref{fig-series1}. This network has two consecutive links: one between $S$ and $R$ and another between $R$ and $T$, each with connection probabilities $p_{SR}$ and $p_{RT}$, respectively.  The overall probability of connecting $S$ and $T$, denoted as $p_{ST}$, is the product of $p_{SR}$ and $p_{RT}$, i.e., 
\begin{eqnarray}
\label{eq_p}
    p_{ST} = p_{SR}\cdot p_{RT}.
\end{eqnarray}
This multiplicative property is essential for modeling many realistic, distance-dependent scenarios. For example, in the case of a lossy channel, if the capacity follows an exponential decay over distance $L$, such that 
$p_{SR}=\exp\left(-L_{SR}\right)$ and $p_{RT}=\exp\left(-L_{RT}\right)$, then the total rate becomes $p_{ST}=\exp\left[-\left(L_{SR}+L_{RT}\right)\right]=p_{SR}\cdot p_{RT}$.
Alternatively, one can think of Eq.~\eqref{eq_p} as the series rule for calculating the total resistance in a series resistor network~\cite{random-walk-electr-netw}, if we define the resistance of each resistor to be $\propto -\ln p$.

A similar behavior occurs in the entanglement swapping protocols of 1D QNs, where each link represents a generic bipartite entangled state (either CV or DV, pure or mixed) shared between the two nodes connected by the link.
Consider, again, Fig.~\ref{fig-series1}, where the two links (entangled states) are denoted by $\rho_{SR}$ and $\rho_{RT}$, respectively. 
Entanglement swapping ideally generates an entangled state $\rho_{ST}$ between nodes $S$ and $T$ using only LOCC, while consuming both $\rho_{SR}$ and $\rho_{RT}$. The entanglement
of the final state $\rho_{ST}$ quantified by the measure $E(\rho_{ST})$ is determined by the entanglement of both initial states, $E(\rho_{SR})$ and $E(\rho_{RT})$, as well as the swapping protocol in use. 

\yaqi{
\emph{Entanglement measures with multiplicative property.---}If $E$ satisfies 
    \begin{eqnarray}\label{eq-multiplication}     E(\rho_{ST})=E(\rho_{SR})\cdot E(\rho_{RT}),
    \end{eqnarray}
then for a 1D chain (see Fig.~\ref{fig-seriesN}) of $l$ identical states with entanglement $E$, we have 
$l/\ln{E(\rho_{ST})}=l/\ln{E^l}=\left(\ln E\right)^{-1}$.
In the following, we explore what swapping protocols enable us to find an entanglement measure with such a multiplicative property.}

\emph{Existence.---}\yaqi{Given an existing entanglement measure $E$ which quantifies the entanglement of the  bipartite states in Fig.~\ref{fig-series1}. If there exists a binary function $g(x,y)$ such that the final entanglement is expressed as $E(\rho_{ST})= g\left(E(\rho_{SR}), E(\rho_{RT})\right)$, does there exist an entanglement measure which can be written as $f_{E}(\rho):=f\left(E(\rho)\right)\in[0,1]$ of $E$, where $f$ is a function satisfying the conditions in Definition~\ref{def1}, such that the equation 
\begin{eqnarray}\label{eq-existence}
    f_{E}\left(\rho_{ST}\right)= f_{E}\left(\rho_{SR}\right)\cdot f_{E}\left(\rho_{RT}\right)
\end{eqnarray}
holds true? Due to the properties of $E$, it can be concluded that the function $g$ must satisfy the following conditions:}
\yaqi{\begin{enumerate} 
    \item $g\left(E(\rho_{SR}), E(\rho_{RT})\right)$ is strictly increasing with respect to both $E(\rho_{SR})$ and $E(\rho_{RT})$, and $g\left(E(\rho_{SR}), E(\rho_{RT})\right)=0$ if and only if $E(\rho_{SR})=0$ or $E(\rho_{RT})=0$;
    \item $g\left(E(\rho_{SR}), E(\rho_{RT})\right)\leqslant\min\{E(\rho_{SR}), E(\rho_{RT})\}$.
\end{enumerate}
Generally, for each of $\rho_{SR}$ and $\rho_{RT}$, all possible values of the entanglement form an interval. 
Based on this, the necessary and sufficient condition for the existence of an entanglement measure $f_{E}$ with the multiplicative property is given as follows: 
\begin{theorem}\label{theorem-multiplication}
    There exists a strictly monotonic function $f$ such that Eq.~\eqref{eq-existence} holds if and only if the function $g$ is a continuous group operation~\cite{Continuous_group1960,Continuous_group_th1966}.
\end{theorem}
We prove this result in Appendix~\ref{sec-AP-multiplication}.}
In particular, if $g\left(E(\rho_{SR}), E(\rho_{RT})\right)=E(\rho_{SR})\cdot E(\rho_{RT})$ already satisfies the multiplicative property, then every power-law function $f(E(\rho))=E^\alpha(\rho)\in[0,1]$ for $\alpha>0$ also satisfies Eq.~\eqref{eq-existence}, reflecting a form of ``gauge redundancy''~\cite{Pecplation_ConPT2021}.

\emph{DV \yaqi{entanglement measures with multiplicative property}.---}For several DV-based entanglement swapping protocols, the corresponding entanglement measures with the multiplicative property already exist and are widely explored. We present a few examples:

\emph{Example 1.---}Consider a 1D QN where each link represents a pure, partially entangled $2$-qubit state [i.e.,~a $(2\times 2)$-dimensional system] which is not maximally entangled.
The entanglement of each link can be characterized by a probability measure $P_{\mathrm{SCP}}$, which is the \yaqi{maximum} probability to convert the state to a singlet (that is, a Bell state up to unitary transformation) by LOCC, namely the singlet conversion probability (SCP)~\cite{Percolation_CEP2007}.
Motivated by Eq.~\eqref{eq_p}, given the probabilities $P_{\mathrm{SCP}}(\rho_{SR})$ and $P_{\mathrm{SCP}}(\rho_{RT})$ of converting states $\rho_{SR}$ and $\rho_{RT}$ to singlets, the probability that both states are converted into singlets is given by
\begin{eqnarray}
    P_{\mathrm{SCP}}(\rho_{ST})=P_{\mathrm{SCP}}(\rho_{SR})\cdot P_{\mathrm{SCP}}(\rho_{RT}).
\end{eqnarray}
Here, $P_{\mathrm{SCP}}(\rho_{ST})$ can be regarded as a SCP (although not the optimal SCP~\cite{percolation2008}) for obtaining a singlet between $S$ and $T$. Once two singlets connected in series are obtained, a trivial Bell-state measurement~\cite{swapping1993} on $R$ can establish a singlet between $S$ and $T$, while consuming the original two singlets, $S$--$R$ and $R$--$T$.
This type of swapping protocol is actually the building block for the classical entanglement percolation proposed in Ref.~\cite{Percolation_CEP2007}, and we see that the SCP shows \yaqi{the multiplicative property} when considering the corresponding swapping protocol. 

\emph{Example 2.---}Consider a 1D QN where each link represents a pure, partially entangled 2-qubit state.
For the deterministic swapping protocol of Ref.~\cite{percolation2008}, the concurrence $C$ exhibits the multiplicative property:
\begin{eqnarray}
\label{eq_example2}
    C(\rho_{ST})=C(\rho_{SR})\cdot C(\rho_{RT}).
\end{eqnarray}
This shows that from the two partially entangled pure qubit states $\rho_{SR}$ and $\rho_{RT}$, one can deterministically obtain a new, partially entangled pure qubit state $\rho_{ST}$ that has concurrence $C(\rho_{ST})$.
This type of swapping protocol is the building block for the concurrence percolation proposed in Ref.~\cite{Pecplation_ConPT2021}.

\emph{Example 3.---}Consider a 1D QN where each link represents a pure, partially entangled 2-qudit state [i.e.,~a $(d\times d)$-dimensional system].
For the deterministic swapping protocol of Ref.~\cite{G_concurrence2005}, the $G$-concurrence~\cite{G_concurrence2005} exhibits the multiplicative property:
\begin{eqnarray}
\label{eq_example3}
    C_{G}(\rho_{ST})=C_{G}(\rho_{SR})\cdot C_{G}(\rho_{RT}).
\end{eqnarray}
Similarly to the previous example, from the two partially entangled pure qudit states $\rho_{SR}$ and $\rho_{RT}$, one can deterministically obtain a new, partially entangled pure qudit state $\rho_{ST}$ that has $G$-concurrence $C_G(\rho_{ST})$.
Moreover, it has been shown that both deterministic swapping protocols [Eqs.~\eqref{eq_example2}~and~\eqref{eq_example3}] are \emph{optimal}, that is, $C(\rho_{ST})$ [$C_{G}(\rho_{ST})$] is the optimal average concurrence ($G$-concurrence) that can be obtained by any deterministic or probabilistic swapping protocols~\cite{DET2023}.
This implies a fundamental upper bound on entanglement, indicating that it \emph{must} decay exponentially over long distances as the number of QN links increases.

\emph{CV \yaqi{entanglement measures with multiplicative property}.---}CV systems are inherently infinite-dimensional. Therefore, one may consider the $G$-concurrence in the $d\to \infty$ limit. However, as $d\to \infty$, the $G$-concurrence approaches zero, reducing Eq.~\eqref{eq_example3} to a trivial form.
To find \yaqi{a nontrivial entanglement measure} that has the multiplicative property for CV systems, we consider a 1D QN where each link denotes a TMSVS. 
The deterministic swapping protocol for TMSVS, as given by  Ref.~\cite{vanLoock2002}, takes two TMSVSs $\rho_{SR}=|\psi^{r_{SR}}\rangle\langle\psi^{r_{SR}}|$ and $\rho_{RT}=|\psi^{r_{RT}}\rangle\langle\psi^{r_{RT}}|$ as input. The output is also a TMSVS, $\rho_{ST}=|\psi^{r_{ST}}\rangle\langle\psi^{r_{ST}}|$, for which the squeezing parameter is determined by $\tanh r_{ST}=\tanh r_{SR} \tanh r_{RT}$~\cite{vanLoock2002}.
Thus, from Eq.~\eqref{ER-r}, it is straightforward to see that the ratio negativity of the final TMSVS $\rho_{ST}$ satisfies 
\begin{eqnarray}
\label{eq_characteristic_measure_chi}
    \chi_{\mathcal{N}}(\rho_{ST})=\chi_{\mathcal{N}}(\rho_{SR})\cdot \chi_{\mathcal{N}}(\rho_{RT}),
\end{eqnarray}
and hence $\chi_{\mathcal{N}}$ exhibits \yaqi{the multiplicative property} when considering the Gaussian swapping protocol for TMSVSs.

\yaqi{Note that the $\alpha$-ratio negativity also exhibits the multiplicative property for the Gaussian swapping protocol, since} Eq.~\eqref{eq_characteristic_measure_chi} could also be rewritten as
\begin{eqnarray}\label{eq-chi_rate}
    \chi_{\mathcal{N}}^{\alpha}(\rho_{ST})=\chi_{\mathcal{N}}^{\alpha}(\rho_{SR})\cdot \chi_{\mathcal{N}}^{\alpha}(\rho_{RT})
\end{eqnarray}
for $\alpha>0$. \yaqi{This measure yields the characteristic length $\xi\propto -\left(\ln \chi_{\mathcal{N}}\right)^{-1}$.}
Additionally, it has been shown that this swapping protocol is optimal, in the sense that the entanglement (the ratio negativity) of the final TMSVS $\rho_{ST}$ is the maximum that can be obtained by any Gaussian swapping protocols~\cite{vanLoock2002}. 
In summary, the various entanglement swapping protocols, as well as their corresponding \yaqi{entanglement measures and correlation lengths}, are shown in Table~\ref{table-measures}. 

\section{Discussion and Conclusions}\label{sec-conclusion}

In several practical applications aiming at implementing quantum technologies, CV states are used as a valid alternative of single photon sources. Whereas the latter offer higher performance, the former are much more feasible with the state-of-the-art devices. Among the various CV states a fundamental role is played by Gaussian states, such as the TMSVSs. These states can be used to perform entanglement swapping operations on a 1D QN. To characterize the properties of entanglement swapping for such a system, we searched for a suitable entanglement measure.
To this aim, we initially identified a family of entanglement measures termed as the $f$-negativity. Among them, the ratio negativity is of special interest, as it is a bounded and nonconvex entanglement monotone with desired properties, such as the subadditivity and the monogamy. We show that the ratio negativity is naturally related to the explicit expression of TMSVSs and can help us determine the correlation length of the optimal Gaussian entanglement swapping protocol for distributing TMSVSs on 1D QNs.

The ratio negativity introduced in this work is a novel figure of merit that can find several potential applications in a variety of quantum systems. For instance, the \emph{optimal} swapping protocol of generic pure CV states (beyond TMSVSs) is not known in general. It is also unclear whether such a protocol can be characterized by an \yaqi{entanglement measure that satisfies} the form of Eq.~\eqref{eq_characteristic_measure},
or bears a more complicated form (akin to the optimal SCP protocol for a 1D QN~\cite{percolation2008}).
We believe that the ratio negativity might find an important role in addressing these issues and a more in-depth analysis of it might unveil new properties of QNs, in which entanglement swapping plays a fundamental role.

\yaqi{In practical experimental settings, one may encounter 1D TMSVS QNs with untrusted nodes~\cite{Untrust_scenario2020,Untrust_scenario2022}. 
Even in specific untrusted situations, the ratio negativity continues to be an entanglement measure with the multiplicative property. For example, when the entanglement swapping protocol with untrusted nodes affects only local displacements, rather than the correlation matrices of the bipartite Gaussian states distributed over the elementary links, the information in those matrices remains invariant~\cite{SP2005,CSRNTJS2012}. 
As a result, the ratio negativity retains its multiplicative property for the CV entanglement swapping protocol discussed above.}

\section*{Acknowledgements}

\yaqi{We thank Chao Zu from the Department of Mathematical Sciences, Dalian University of Technology for  help and useful discussions.} K.H. and J.H. were supported by the National Natural Science Foundation of China under Grants No.~12271394 and No.~12071336.
N.L.P. was supported by the KAKENHI project 24K07485.

\newpage
\clearpage

\section*{APPENDIX}
\appendix

\renewcommand{\appendixname}{Appendix}
\renewcommand{\thesection}{\Roman{section}}
\renewcommand{\theequation}{\Roman{section}.\arabic{equation}}
\renewcommand{\thesubsection}{\Alph{subsection}}

\section{Covariant matrix and Gaussian states}\label{sec-AP-CM}

Denote by $\hat{a}_{k}$ the \emph{annihilation operator}, defined by $\hat{a}_{k}|n\rangle=\sqrt{n}|n-1\rangle$ for $n>0$ and $\hat{a}_{k}|0\rangle=0$ acting on the Hilbert space ${\mathcal H}_{k}$; $\hat{a}_{k}^{\dag}$ is the \emph{creation operator} that is the conjugate of $\hat{a}_{k}$. The \emph{position and momentum operators} are defined as $\hat{q}_{k}={\frac{1}{\sqrt{2}}}(\hat{a}_{k}+\hat{a}_{k}^{\dag})$,   $\hat{p}_{k}=\frac{1}{\sqrt{2}i}(\hat{a}_{k}-\hat{a}_{k}^{\dag})$, respectively.
Let $\hat{R}= \left(\hat{R}_1,\hat{R}_2,\dots,\hat{R}_{2N}\right)^{T}=(\hat{q}_{1}, \hat{p}_{1}, \ldots, \hat{q}_{N}, \hat{p}_{N})^{T}$, and let $\Lambda=\oplus_{k=1}^{N}J$ with $J=
  \begin{pmatrix}
       0 & 1\\
       -1 & 0
  \end{pmatrix}$.
  
The \emph{covariant matrix} of a state $\sigma\in \mathbb{S}(\mathcal H)$ is a $2N\times2N$ real symmetric matrix $\Gamma=\Gamma_{sl}$ defined by
  \begin{eqnarray}
    \label{CMsl}
    \Gamma_{sl}:=2(\langle{\hat{R}}_{s}{\hat{R}}_{l}\rangle-\langle{\hat{R}}_{s}\rangle\langle{\hat{R}}_{l}\rangle)-i\Lambda_{sl}
  \end{eqnarray}
for $s,l=1,2,\dots,2N$,
where $\langle \hat{A}\rangle:=\Tr(\sigma \hat{A})$ is the expectation of operator $\hat{A}$ on the density matrix $\sigma$.
The \emph{characteristic function} of $\sigma$ is defined as $\chi(\xi):=\langle\mathcal W(\xi)\rangle$ for $\xi\in\mathbb{R}^{2N}$, where ${\mathcal W}(\xi):=\exp{(i\xi^{T}\hat{R})}$ is referred to as the \emph{Weyl operator}.

A \emph{Gaussian state} $\rho$ is the state whose characteristic function is Gaussian:
\begin{eqnarray}
    \chi(\xi)=\exp{\left(-\frac{1}{4}\xi^{T}\Gamma\xi+id^{T}\xi\right)}.
\end{eqnarray}
Here, $d\in\mathbb{R}^{2N}$ is referred to as the \emph{displacement vector}, defined by $d_{j}:=\langle\hat{R}_{j}\rangle$.

\section{Properties of the $f$-negativity}\label{sec-AP-properties}

\begin{proposition}
  $f_{\mathcal{N}}$ is an entanglement measure.
\end{proposition}
\begin{proof}
For an arbitrary LOCC map $\Phi$, by Eq.~\eqref{equivalence}, one has $f_{\mathcal{N}}(\rho)\geqslant f_{\mathcal{N}}(\Phi(\rho))$, given that $\mathcal{N}$ is an entanglement measure itself.
Therefore, the $f$-negativity does not increase under any LOCC map $\Phi$.
This result combining with the semipositivity assures that the $f$-negativity is an entanglement measure.  
\end{proof}

\begin{proposition}
  $f_{\mathcal{N}}$ is convex when the function $f$ is convex.
\end{proposition}
\begin{proof}
Let the function $f$ be convex.
Since the negativity $\mathcal{N}$ is a convex function, by Eq.~\eqref{equivalence}, we have
    \begin{eqnarray*}
        &&f_{\mathcal{N}}\left(\sum_j q_j\rho_j\right)\nonumber\\
        &=&f\left(\mathcal{N}\left(\sum_j q_j\rho_j\right)\right)
        \leqslant f\left(\sum_j q_j\mathcal{N}(\rho_j)\right).
    \end{eqnarray*}
    If $f$ is convex, then 
    \begin{eqnarray*}
        f\left(\sum_j q_j\mathcal{N}(\rho_j)\right)
        \leqslant \sum_j q_j f\left(\mathcal{N}(\rho_j)\right)
        =\sum_j q_j f_{\mathcal{N}}(\rho_j),
    \end{eqnarray*}
    implying that $f_{\mathcal{N}}\left(\sum_j q_j\rho_j\right)\leqslant\sum_j q_j f_{\mathcal{N}}(\rho_j)$, implying that the $f$-negativity is convex.
\end{proof}

\emph{Remark.---}Conversely, for a nonconvex function $f$, the convexity of the $f$-negativity is not guaranteed. We give a counterexample.
Suppose that $f(x):=\sqrt{x}$,  we
give three states as
\begin{equation}
  \begin{aligned}
    \rho_1&=\frac{1}{2}(|00\rangle\langle00|+|11\rangle\langle11|),\\ \rho_2&=|\phi^{+}\rangle\langle\phi^{+}|, \\
    \rho&=\frac{1}{2}(\rho_1+\rho_2)\\
  \end{aligned}
\end{equation}
where 
\begin{eqnarray}\label{eq-singlet}
    |\phi^{+}\rangle=\frac{|00\rangle+|11\rangle}{\sqrt{2}}
\end{eqnarray}
is a Bell state. 
It follows that  $\mathcal{N}(\rho)=1/4$, $\mathcal{N}(\rho_1)=0$ and $\mathcal{N}(\rho_2)=1/2$, implying that $f_{\mathcal N}(\rho)=1/2$ and 
$\left[f_{\mathcal{N}}(\rho_1)+f_{\mathcal{N}}(\rho_2)\right]/2=\left[f(\mathcal{N}(\rho_1)+f(\mathcal{N}(\rho_2))\right]/2=
\left[f(0)+f\left(\frac{1}{2}\right)\right]/2=f\left(1/2\right)/2=\sqrt{2}/{4}<1/2=f_{\mathcal N}(\rho)$. It follows that  the $f$-negativity $f_{\mathcal{N}}$ is nonconvex. 

Indeed, similar to the aforementioned example, one can calculate and show that the $f$-negativity is nonconvex for the following functions,  $f(x):= x^{\xi}$ for $0<\xi<1$, $f(x):= [\ln{(2x+1)}]^{\eta}$ for $0<\eta<\log_{2}(\log_{3/2}2)\approx 1.293$ and $f(x):=\left[x/\left({x+1}\right)\right]^{\alpha}$ for $\alpha<\log_{5/3}2\approx 1.357$.

\begin{proposition}
  $f_{\mathcal{N}}$ is an entanglement monotone if $f$ is concave (including linear) in $[0,+\infty)$.
\end{proposition}
\begin{proof}
Consider that the function $f$ is concave.
Let $\{p_{i},\rho_{i}'\}$  be the ensemble of the final state obtained after state $\rho$ undergoes LOCC. Since the negativity $\mathcal{N}(\rho)$ is an entanglement monotone~\cite{Vidal2002}, we have
$\mathcal{N}(\rho)\geqslant\sum_{i}p_{i}
   \mathcal{N}(\rho_{i}')$ which means $f_{\mathcal{N}}(\rho)
   =f\left(\mathcal{N}(\rho)\right)
   \geqslant f\left(\sum_{i}p_{i}
   \mathcal{N}(\rho_{i}')\right)
   \geqslant \sum_{i}p_{i}
   f\left(\mathcal{N}(\rho_{i}')\right) 
   =\sum_{i}p_{i} f_{\mathcal{N}}(\rho_{i}')$ 
holds due to its monotonicity and concavity. Thus, the $f$-negativity is an entanglement monotone. 
\end{proof}

\emph{Remark.---}Conversely, not every $f$-negativity is an entanglement monotone. We give a counterexample.
Let $f(x):=x^4$ which is convex.
For a bipartite pure state $|\psi\rangle_{AB}=\sqrt{0.1}|00\rangle+\sqrt{0.9}|11\rangle$, we perform local POVM $\mathcal{M}=\{M_1,M_2\}$ on party $A$ where $M_1=\sqrt{0.8}|0\rangle_A\langle0|+\sqrt{0.2}|1\rangle_A\langle1|$ and 
$M_2=\sqrt{0.2}|0\rangle_A\langle0|+\sqrt{0.8}|1\rangle_A\langle1|$.
Then we obtain that 
\begin{eqnarray}
    f_{\mathcal{N}}=0.0081&<&p_1 f_{\mathcal{N},1}+p_2 f_{\mathcal{N},2}\nonumber\\
    &=&\frac{13}{50}\times \left(\frac{6}{13}\right)^4+\frac{37}{50}\times \left(\frac{6}{37}\right)^4\nonumber\\
    &\approx&0.0123
\end{eqnarray}
where $f_{\mathcal{N}}$, $f_{\mathcal{N},1}$ and $f_{\mathcal{N},2}$ denote $f$-negativity for input state $|\psi\rangle_{AB}$ and two output states, and $p_1$ and $p_2$ are probabilities to get two output states, respectively.

\begin{proposition}
  $f_{\mathcal{N}}(\rho)=0$ is a necessary and sufficient condition for the separability of all $(1+M)$-mode Gaussian states and $(N+M)$-mode bisymmetric Gaussian states.
\end{proposition}
\begin{proof}
By definition, $f_{\mathcal{N}}(\rho)=0$ iff $\mathcal{N}(\rho)=0$. When a Gaussian state $\rho$ is $(1+M)$-mode or $(N+M)$-mode bisymmetric, the vanishing of the negativity (or logarithmic negativity) $\mathcal{N}(\rho)=0$ is a necessary and sufficient condition for the separability of $\rho$~\cite{1Mmode_Gaussian2001,NMmode_bisymmetric_Gaussian2005}, and thus the proposition naturally follows.
\end{proof}

\section{Proof of subadditivity}\label{sec-AP-tensor}

Here, we first show an examples of the ratio negativity applied to the tensor product of quantum states widely used in QNs.

Given Hilbert spaces $\mathcal H_{A}=\otimes_{k=1}^{K}\mathcal H_{A_{k}}$, $\mathcal H_{B}=\otimes_{k=1}^{K}\mathcal H_{B_{k}}$, $\mathcal H_{A_{k}B_{k}}=\mathcal H_{A_{k}}\otimes\mathcal H_{B_{k}}$, and $\rho_{k}\in \mathbb{S}(\mathcal{H}_{A_{k}B_{k}})$ $(k=1,2,\ldots,K)$, for the state $\rho=\otimes_{k=1}^{K}\rho_{k}$, we have $\rho\in \mathbb{S}(\mathcal{H}_{AB})$ and
  \begin{eqnarray}\label{eq-tensorN}
         \chi_{\mathcal N}(\rho)
          =\frac{\prod\limits_{k=1}^{K}\frac{1+\chi_{\mathcal N}(\rho_{k})}{1-\chi_{\mathcal N}(\rho_{k})}-1}
         {\prod\limits_{k=1}^{K}\frac{1+\chi_{\mathcal N}(\rho_{k})}{1-\chi_{\mathcal N}(\rho_{k})}+1}.
  \end{eqnarray}

For the case of $K=2$, from Eq.~\eqref{eq-tensorN}, we have
  \begin{eqnarray}\label{ineq-2}
   \begin{aligned}
      \chi_{\mathcal N}(\rho)=\frac{\chi_{\mathcal N}(\rho_{1})+\chi_{\mathcal N}(\rho_{2})}{1+\chi_{\mathcal N}(\rho_{1})\chi_{\mathcal N}(\rho_{2})}\\
      \leqslant \chi_{\mathcal N}(\rho_{1})+\chi_{\mathcal N}(\rho_{2})\ \
   \end{aligned}
  \end{eqnarray}
where the equality is true only when $\chi_{\mathcal N}(\rho_{1})=0$ or $\chi_{\mathcal N}(\rho_{2})=0$.
Recall that a state $\rho_{AB}$ of a bipartite system $AB$ is called the positive partial transpose (PPT) state if $\rho^{T_A}$ is positive.
In addition,  $\rho_{AB}$ is a PPT state iff $\mathcal{N}(\rho_{AB})=0$~\cite{Vidal2002}.
From the ordering of the ratio negativity, we then have that the equality in Eq.~\eqref{ineq-2} is ture iff $\rho_{1}$ or $\rho_{2}$ is a PPT state.

\section{Monogamy inequalities}\label{sec-AP-monogamy}

\subsection{$\chi_{\mathcal{N}}$ violates the CKW inequality}\label{sec-AP-CKW}
Consider, for example, $\rho_{ABC}=|\psi\rangle_{ABC}\langle\psi|\in\mathbb{PS}(\mathcal{H}_{ABC})$ with $|\psi\rangle_{ABC}=(|000\rangle+|011\rangle+\sqrt{2}|110\rangle)/2$. This yields $\chi_{\mathcal{N}}(\rho_{AB})=\chi_{\mathcal{N}}(\rho_{AC})=1/5$ and $\chi_{\mathcal{N}}(\rho_{A|BC})=1/3$, where $\rho_{AB(AC)}$ is the reduced state of $\rho_{ABC}$ in subsystems $AB(AC)$. Thus,
\begin{eqnarray}\label{ineq-monogamy}
    \chi^{\alpha}_{\mathcal{N}}(\rho_{AB})+\chi^{\alpha}_{\mathcal{N}}(\rho_{AC}) \nless \chi^{\alpha}_{\mathcal{N}}(\rho_{A|BC})
\end{eqnarray}
when $\alpha=1$.

\subsection{Monogamy inequalities of $\chi_{\widetilde{\mathcal{N}}}^{\alpha}$}\label{sec-AP-extended-ratio-monogamy}

\begin{theorem}
  Let us consider an $N$-qubit quantum system defined in an $N$-qubit Hilbert space $\mathcal{H}_{AB_1B_2\cdots B_{N-1}}$ of $N$ parties $A,B_1,B_2,\dots,B_{N-1}$. Then for $\alpha\geqslant \log_{3\left(\sqrt{2}-1\right)}2\approx 3.191$, the inequality 
  \begin{eqnarray}\label{ineq-ratio}
      \chi_{\widetilde{\mathcal{N}}}^{\alpha}(\rho_{A|B_1B_2\cdots B_{N-1}})\ \geqslant\sum_{n=1}^{N-1}\chi_{\widetilde{\mathcal{N}}}^{\alpha}(\rho_{AB_n})
  \end{eqnarray}
  holds true for $\rho_{AB_1B_2\cdots B_{N-1}}\in\mathbb{S}(\mathcal{H}_{AB_1B_2 \cdots B_{N-1}})$, where $\rho_{AB_n}$ is the reduced state of $\rho_{AB_1B_2\cdots B_{N-1}}$ in subsystem $AB_n$. 
\end{theorem}
\begin{proof}
Inspired by Ref.~\cite{Monogamy2021}, before proving Eq.~\eqref{ineq-ratio}, we need to show that in the region $D=\{(x,y)|0\leqslant x\leqslant a,0\leqslant y\leqslant b, a\leqslant b\}$, the inequality
    \begin{eqnarray}\label{ineq-xya}
        \left[\frac{x}{x+1}\right]^{\alpha}+\left[\frac{y}{y+1}\right]^{\alpha}
        \leqslant\left[\frac{\sqrt{x^2+y^2}}{\sqrt{x^2+y^2}+1}\right]^{\alpha}
    \end{eqnarray}
holds for $\alpha\geqslant g\left(c,b/c\right)$ where $c\equiv\sqrt{a^2+b^2}$, and the  auxiliary function $g$ is defined by
\begin{eqnarray}
    g(r,u):= 2\log_{\frac{(1+r)u}{1+ru}}
    u.
\end{eqnarray}

To prove Eq.~\eqref{ineq-xya}, we first need to show that in the region $D'=\{(r,u)|0\leqslant ru\leqslant a, 0\leqslant r\sqrt{1-u^2}\leqslant b,0\leqslant u\leqslant1\}$, the inequality
\begin{eqnarray}\label{ineq-g1}
    \left[\frac{(1+r)u}{1+ru}\right]^{\alpha}\leqslant u^2
\end{eqnarray}
holds for $\alpha\geqslant g\left(c,a/c\right)$.
To see this, since the Eq.~\eqref{ineq-g1} is true iff $\alpha\geqslant g(r,u)$ is true in $D'$, we only need to prove $\alpha\geqslant \max\limits_{(r,u)\in D'}g(r,u)$.
Due to the fact that $\partial g(r,u)/\partial r>0$ holds in $D'$, we obtain 
\begin{equation}\label{eq-max_g}
   \max\limits_{(r,u)\in D'}g(r,u)=\max\limits_{(c,u)\in D'}g(c,u)=g\left(c,a/c\right).
\end{equation}
Thus for all $\alpha\geqslant g\left(c,a/c\right)$, Eq.~\eqref{ineq-g1} holds in  $D'$ from Eq.~\eqref{eq-max_g}.

Now we prove Eq.~\eqref{ineq-xya}: 
For any point $(x,y)$ in the region $D=\{(x,y)|0\leqslant x\leqslant a,0\leqslant y\leqslant b\}$, we convert it into polar coordinates, $r=\sqrt{x^2+y^2}$, $\cos\theta=x/r$
    and $\sin\theta=y/r$. Let $c=\sqrt{a^2+b^2}$, then the region $D$ can be rewritten as 
    \begin{eqnarray*}
        D''=\{(r,\theta)|0\leqslant r\cos\theta\leqslant a,0\leqslant r\sin\theta\leqslant b, 0\leqslant r\leqslant c\}.
    \end{eqnarray*}
    Then, Eq.~\eqref{ineq-xya} holds iff the inequality 
\begin{eqnarray}\label{ineq-leq}
    \left[\frac{(1+r)\cos{\theta}}{1+r\cos{\theta}}\right]^{\alpha}+\left[\frac{(1+r)\sin{\theta}}{1+r\sin{\theta}}\right]^{\alpha}\leqslant1
\end{eqnarray}
holds in $D''$. 

It is sufficient to prove that
\begin{eqnarray}\label{ineq-cos}
    \left[\frac{(1+r)\cos{\theta}}{1+r\cos{\theta}}\right]^{\alpha}\leqslant\cos^2\theta
\end{eqnarray}
and 
\begin{eqnarray}\label{ineq-sin}
    \left[\frac{(1+r)\sin{\theta}}{1+r\sin{\theta}}\right]^{\alpha}\leqslant\sin^2\theta
\end{eqnarray}
are true in $D''$.
By replacing $u$ in Eq.~\eqref{ineq-g1} with $\cos \theta$, we then have that Eq.~\eqref{ineq-cos} holds true in $D''$ when 
\begin{eqnarray}
    \alpha\geqslant g\left(c,a/c\right).
\end{eqnarray}
Let $u$ be $\sin \theta$ and swap the values of $a$ and $b$, we obtain that Eq.~\eqref{ineq-sin} holds true for all $(r,\theta)\in D''$ when
\begin{eqnarray}
     \alpha\geqslant g\left(c,b/c\right).
\end{eqnarray}
Since $\partial g(r,u)/\partial u\geqslant 0$ and $a\leqslant b$, Eq.~\eqref{ineq-xya} holds true in $D''$ for all $\alpha\geqslant g\left(c,b/c\right)$. 
We complete the proof.

At the end, we prove Eq.~\eqref{ineq-ratio}.
Denote by $\mathcal{H}_{AB}$ an $N$-qubit system shared between $N$ parities $A,B_1,B_2,\dots, B_{N-1}$, and let $B=\{B_1,B_2,\dots, B_{N-1}\}$ denote the collection of $N-1$ parties except $A$.
For $\rho_{AB}\in\mathbb{S}(\mathcal{H}_{AB})$, we denote by $\rho_{AB_n}$ the reduced state of $\rho_{AB}$ in subsystem $AB_n$.
The negativity is known to satisfy the monogamy inequality~\cite{Kim2018,Monogamy2021}
\begin{eqnarray}\label{ineq-N2}
    \widetilde{\mathcal{N}}^2(\rho_{A|B})\geqslant \sum_{n=1}^{N-1}\widetilde{\mathcal{N}}^2(\rho_{AB_n}),
\end{eqnarray}
which implies that
\begin{equation}\label{ineq-XNa}
\begin{small}
 \begin{aligned}
    \chi_{\widetilde{\mathcal{N}}}^{\alpha}(\rho_{A|B})
    &=\left[\frac{\sqrt{\widetilde{\mathcal{N}}^2(\rho_{AB})}}{\sqrt{\widetilde{\mathcal{N}}^2(\rho_{AB})}+1}\right]^{\alpha}\\
    &\geqslant\left[\frac{\sqrt{\sum\limits_{n=1}^{N-1}\widetilde{\mathcal{N}}^2(\rho_{AB_n})}}{\sqrt{\sum\limits_{n=1}^{N-1}\widetilde{\mathcal{N}}^2(\rho_{AB_n})}+1}\right]^{\alpha}
 \end{aligned}
 \end{small}
\end{equation}
holds for all $\alpha>0$, a result of the monotonicity of the function $f(x)= [x/(x+1)]^{\alpha}$.

Note that every state of $(2\times d)$-dimensional systems for $d\geqslant 2$ has the negativity no greater than $1/2$~\cite{Vidal2002}. 
Therefore, from Eq.~\eqref{ineq-N2}, we obtain
\begin{eqnarray}
    \sum_{n=1}^{N-1}\widetilde{\mathcal{N}}^2(\rho_{AB_n})\leqslant \frac{1}{2}.
\end{eqnarray}
Let the parameters in Eq.~\eqref{ineq-xya} be $x=\widetilde{\mathcal{N}}(\rho_{AB_1})$, $y=\sqrt{\sum_{n=2}^{N-1}\widetilde{\mathcal{N}}^2(\rho_{AB_n})}$, $a=1/2$, $b=1/2$ and $c=1/2$.
We find that the inequality 
\begin{equation}\label{ineq-N1}
  \begin{aligned}
    &\left[\frac{\sqrt{\sum\limits_{n=1}^{N-1}\widetilde{\mathcal{N}}^2(\rho_{AB_n})}}{\sqrt{\sum\limits_{n=1}^{N-1}\widetilde{\mathcal{N}}^2(\rho_{AB_n})}+1}\right]^{\alpha}\\
    \geqslant& \left[\frac{\widetilde{\mathcal{N}}(\rho_{AB_1})}{\mathcal{N}(\rho_{AB_1})+1}\right]^{\alpha}
    +\left[\frac{\sqrt{\sum\limits_{n=2}^{N-1}\widetilde{\mathcal{N}}^2(\rho_{AB_n})}}{\sqrt{\sum\limits_{n=2}^{N-1}\widetilde{\mathcal{N}}^2(\rho_{AB_n})}+1}\right]^{\alpha}
  \end{aligned}
\end{equation}
holds for $\alpha\geqslant g\left(1/{\sqrt{2}},1/{\sqrt{2}}\right)= \log_{3\left(\sqrt{2}-1\right)}2$. 
By iteratively applying Eq.~\eqref{ineq-N1} $N-2$ times, we find that the inequality
\begin{equation}\label{ineq-XNa2}
 \begin{aligned}
    \left[\frac{\sqrt{\sum\limits_{n=1}^{N-1}\widetilde{\mathcal{N}}^2(\rho_{AB_n})}}{\sqrt{\sum\limits_{n=1}^{N-1}\widetilde{\mathcal{N}}^2(\rho_{AB_n})}+1}\right]^{\alpha}
    \geqslant&\sum_{n=1}^{N-1}\left[\frac{\widetilde{\mathcal{N}}(\rho_{AB_n})}{\mathcal{N}(\rho_{AB_n})+1}\right]^{\alpha}\\
    =&\sum_{n=1}^{N-1}\chi_{\widetilde{\mathcal{N}}}^{\alpha}(\rho_{AB_n})
 \end{aligned}
\end{equation}
holds for $\alpha\geqslant \log_{3\left(\sqrt{2}-1\right)}2$.
From Eqs.~\eqref{ineq-XNa} and~\eqref{ineq-XNa2}, we thus obtain the monogamy inequality
    \begin{eqnarray}\label{ineq-monogamyN}
        \chi_{\widetilde{\mathcal{N}}}^{\alpha}(\rho_{A|B})\geqslant \sum_{n=1}^{N-1}\chi_{\widetilde{\mathcal{N}}}^{\alpha}(\rho_{AB_n})
    \end{eqnarray}
for $\alpha\geqslant \log_{3\left(\sqrt{2}-1\right)}2$. 
\end{proof}

\begin{theorem}
  Let a tripartite system $\mathcal{H}_{AB_1B_2}$ be $(2\times2\times 3)$-dimensional or $(2\times2\times2^m)$-dimensional. For $\rho_{AB_1B_2}\in\mathbb{S}(\mathcal{H}_{AB_1B_2})$, we denote $\rho_{AB_1(AB_2)}$ as its reduced state in subsystem $AB_1(AB_2)$. Then, similarly, the inequality 
  \begin{eqnarray}\label{ineq-tripartite1}
    \chi_{\widetilde{\mathcal{N}}}^{\alpha}(\rho_{A|B_1 B_2})
        &\geqslant\chi_{\widetilde{\mathcal{N}}}^{\alpha}(\rho_{AB_1})+\chi_{\widetilde{\mathcal{N}}}^{\alpha}(\rho_{AB_2})
  \end{eqnarray}
  holds for all $\alpha\geqslant \log_{3\left(\sqrt{2}-1\right)}2\approx 3.191$.
\end{theorem}
\begin{proof}
Let a tripartite system $\mathcal{H}_{AB_1 B_2}$ be $(2\times2\times 3)$-dimensional or $(2\times2\times2^m)$-dimensional.
For all $\rho_{AB_1 B_2}\in\mathbb{S}(\mathcal{H}_{AB_1 B_2})$,
Ref.~\cite{Monogamy2021} shows that the inequality
\begin{eqnarray}\label{ineq-CREN}
    \widetilde{\mathcal{N}}^2(\rho_{A|B_1 B_2})
        \geqslant
        \widetilde{\mathcal{N}}^{2}(\rho_{AB_1})+\widetilde{\mathcal{N}}^{2}(\rho_{AB_2})
\end{eqnarray}
holds where $\rho_{AB_1(AB_2)}$ is the reduced state of $\rho_{AB_1 B_2}$ in subsystem $AB_1(AB_2)$. 
By iterating Eq.~\eqref{ineq-xya} together with Eq.~\eqref{ineq-CREN}, we find that the monogamous inequality~\eqref{ineq-tripartite1}
holds for all $\alpha\geqslant \log_{3\left(\sqrt{2}-1\right)}2$ because of the monotonicity of function $f(x)=\left[x/\left(x+1\right)\right]$ and the convexity of $\widetilde{\mathcal{N}}$.
\end{proof}

\subsection{Monogamy inequalities of $\chi_{\mathcal{N}}^{\alpha}$}\label{sec-AP-ratio-monogamy}

It is straightforward to see that $\chi_{\widetilde{\mathcal{N}}}^{\alpha}$ is an entanglement measure that satisfies
\begin{eqnarray}\label{eq-relation}
    \chi_{\widetilde{{\mathcal{N}}}}^{\alpha}(\rho_{AB})=\chi_{\mathcal{N}}^{\alpha}(\rho_{AB})
\end{eqnarray}
 for every pure state $\rho_{AB}\in\mathbb{PS}(\mathcal{H}_{AB})$.
Since $\widetilde{\mathcal{N}}(\rho_{AB})\geqslant\mathcal{N}(\rho_{AB})$ by the convexity of $\mathcal{N}$~\cite{Vidal2002},
we obtain
\begin{eqnarray}\label{ineq-RCREN-X}
    \chi_{\widetilde{\mathcal{N}}}^{\alpha}(\rho_{AB})\geqslant\chi_{\mathcal{N}}^{\alpha}(\rho_{AB})
\end{eqnarray}
for a general state $\rho_{AB}\in\mathbb{S}(\mathcal{H}_{AB})$.

Let the $N$-party system $\mathcal{H}_{AB_1B_2\cdots B_{N-1}}$ be the system of $N$-qubit, or $N=3$ with $(2\times 2\times 3)$-dimensional or $(2\times 2\times 2^m)$.
If $\rho_{AB_1B_2 \cdots B_{N-1}}=|\psi\rangle_{AB_1B_2\cdots B_{N-1}}\langle \psi|$ is a pure state, by Eqs.~\eqref{eq-relation}~and~\eqref{ineq-RCREN-X}, the $\alpha$-ratio negativity satisfies
\begin{eqnarray}\label{ineq-convex_roof}
\chi^{\alpha}_{\mathcal{N}}(\rho_{A|B_1B_2\cdots B_{N-1}})
&=& \chi^{\alpha}_{\widetilde{\mathcal{N}}}(\rho_{A|B_1B_2\cdots B_{N-1}})\nonumber\\
    &\geqslant&\sum_{n}\chi^{\alpha}_{\widetilde{\mathcal{N}}}(\rho_{AB_n})
    \nonumber\\
    &\geqslant&\sum_{n}\chi^{\alpha}_{\mathcal{N}}(\rho_{AB_n}).
\end{eqnarray}
Hence, for pure states specifically, the $\alpha$-ratio negativity $\chi_{\mathcal{N}}^{\alpha}$ shows monogamy for $N$-qubit systems, as well as $(2\times2\times 3)$-dimensional and $(2\times2\times2^m)$-dimensional tripartite systems.

\section{Proof of Theorem~\ref{theorem-multiplication}}\label{sec-AP-multiplication}

{Before proving Theorem~\ref{theorem-multiplication}, we introduce a definition of the continuous group operation and the Theorem~3 in Sec.~2.2.4 in Ref.~\cite{Continuous_group_th1966}:}

{Let $G$ be a non-empty set. Then $G$ is a group under the operation $\phi:\ G\times G\rightarrow G$ if $\phi$ satisfies the following conditions:
\begin{enumerate}[label=\arabic*),itemsep=0pt,topsep=0pt,parsep=0pt]
  \item Closure: For any two elements $x,y\in G$, the result of $\phi(x,y)$ is also in the set $G$;
  \item Associative: For any two elements $x,y,z\in G$, $\phi\left(\phi(x,y),z\right)=\phi\left(x,\phi(y,z)\right)$;
  \item Identity element: There exists an identity element $e$ such that for every element $x\in G$, $\phi(x,e)=\phi(e,x)=x$;
  \item Inverse element: For each element $x\in G$, there exists an element $y\in G$ such that $\phi(x,y)=\phi(y,x)=e$, where $e$ is the identity element.
\end{enumerate}
\begin{theorem}\label{theorem-continuous_group}
(Theorem~3 in Sec.~2.2.4 in~\cite{Continuous_group_th1966})
    If in an equation of the form 
    \begin{eqnarray}\label{eq-theorem3}
        f\left[g(x,y)\right]=h\left[f(x),f(y)\right],
    \end{eqnarray}
    one of the functions $h$, $g$, for example, $h(z,w)$, is a continuous group operation for the $z$, $w$ of an interval, then Eq.~\eqref{eq-theorem3} has a continuous, strictly monotonic solution $f$ if and only if the other function (here $g$) is also a continuous group operation.
\end{theorem}}

{Let $h(z,w):=zw$, $x=E(\rho_{SR})$ and $y=E(\rho_{RT})$. Then the function $h(z,w)$ is a continuous group operation for the variables $z$ and $w$. 
Therefore, from Theorem~\ref{theorem-continuous_group}, we obtain the result shown in Theorem~\ref{theorem-multiplication}.}

\newpage
\clearpage

\bibliography{XN_SI}

\end{document}